\documentclass[11pt]{article}
\usepackage{amsfonts}
\usepackage[margin=2cm]{geometry}
\usepackage{graphicx,color}
\usepackage[hang,small,bf]{caption}
\usepackage{xspace}
\usepackage{enumerate}
\usepackage{amssymb,amsmath}
\usepackage{amsthm}
\usepackage{tkz-graph}

\usepackage[T1]{fontenc}
\usepackage[utf8]{inputenc}
\usepackage{authblk}

\usepackage[linesnumbered,ruled,vlined]{algorithm2e}

\usepackage[utf8]{inputenc}

\usepackage{amsmath,amsfonts}
\usepackage{cite}
\usepackage{graphicx}
\usepackage{todonotes}

\usepackage{tikz}
\usetikzlibrary{shapes}

\newcommand{\V}{V} 
\newcommand{\SSS}{[N]} 
\newcommand{\Objection}{\mathcal{P} } 
\newcommand{\fctV}{v} 

\newcommand{\ContreObjection}{\mathcal{Q} } 
\newcommand{\ShortestPath}{\mathcal{SH}} 
\newcommand{\C}{\mathcal{C}} 
\newcommand{\imputation}{x} 
\newcommand{\price}{p} 

\newcommand{\X}{X} 

\newtheorem{thm}{Theorem}

\newtheorem{corollary}[thm]{Corollary}
\newtheorem{theorem}[thm]{Theorem}

\newtheorem{example}[thm]{Example}
\newtheorem{definition}[thm]{Definition}

\newtheorem{lemma}[thm]{Lemma}

\newcommand{\player}{o}

\usepackage{xspace}
\newcommand{\Org}[1]{\ensuremath{O^{(#1)}}\xspace}

\newcommand{\Job}[1]{\ensuremath{\mathcal{J}^{(#1)}}\xspace}

\newcommand{\ekloc}[1]{\ensuremath{E^{(#1)}_{\text{local}}\xspace}}

\newcommand{\costlocal}[1]{\ensuremath{\text{cost}^{(#1)}_{\text{local}}}\xspace}
\newcommand{\costcoop}[1]{\ensuremath{\text{cost}^{(#1)}}\xspace}
\newcommand{\Cost}[1]{\costcoop{#1}}
\newcommand{\CostS}[1]{C_{\ShortestPath}^{(#1)}}
 
\begin{document}

\title{Detecting service provider alliances on the choreography enactment pricing game\footnote{This work was supported by projet CAPES-COFECUB MA 828-15 CHOOSING}}
\author[1]{Johanne Cohen}
\author[2]{Daniel Cordeiro}
\author[3]{Loubna Echabbi}

\affil[1]{LRI, Universit\'e Paris Sud, France}  
\affil[2]{Universidade de S\~ao Paulo, Brazil}
\affil[1]{STRS Lab., INPT, Morocco}

\maketitle
%
%
\begin{abstract}
  We present the choreography enactment pricing game, a cooperative
  game-theoretic model for the study of scheduling of jobs using
  competitor service providers. A choreography (a peer-to-peer service
  composition model) needs a set of services to fulfill its jobs
  requirements. Users must choose, for each requirement, which service
  providers will be used to enact the choreography at lowest cost. Due
  to the lack of centralization, vendors can form alliances to control
  the market. We show a novel algorithm capable of detecting
  alliances among service providers, based on our study of the
  bargaining set of this game.
\end{abstract}


\section{Introduction}

Modern distributed systems are usually modeled and described by a
service-oriented architecture that is completely
distributed. The ongoing transition from the current
service \emph{orchestration} model for service composition --- where
the ensemble of services are composed as an executable business
process, controlled by a single organization (the orchestrator) --- to
the service \emph{choreography} composition model~\cite{orc-vs-chor}
--- that describes a non-executable protocol for peer-to-peer
interactions among different organizations.

This new model is not only more robust and scalable, but also more
collaborative. The popularization of service choreography model
enforces interoperability and loose coupling by reflecting obligations
and constraints among different parties. This trend can be seem as a
disruptive change on distributed software development, creating new
opportunities for resource sharing among different organizations.

One of the implications of this trend is that even if a user chooses
one service vendor (e.g., a cloud computing provider) to execute its
application, there is nothing that prevents the vendor to subcontract
resources from other vendors and scatter the application services
among those vendors, creating a collaborative platform composed of
resources from different organizations.


Beyond the complexities of the management of composite services
(design, provisioning, etc.), we consider the problem from an economic
point of view. Since service vendors are not regulated, both
cooperative and non-cooperative behavior may be expected. This can
potentially lead to the formation of alliances --- i.e., the formation
of groups of similar independent organizations, who join together to
control prices and/or limit competition.


In this work we use tools and techniques from cooperative game
theory~\cite{courcoubetis} in order to analyze the stability and the
alliance formation on this new scenario; a problem we call the
\emph{choreography enactment pricing game}.

\section{Problem statement}
\label{sec:problem-statement}

In this paper, we study how task assignment between independent
service vendors can lead to the establishment of alliances on
unregulated economies. We call \emph{organization} a vendor, its
physical computational resources and the users subscribed to it.

To state the problem, we will use the standard notation for scheduling
on multi-organization platforms~\cite{mosp-europar}.  More formally,
we have a set of $N$ different organizations. Each organization
\Org{k}, $1 \le k \le N$, has $m^{(k)}$ machines that can be used to
execute jobs submitted by its users or by other organizations.

A user must choose an organization to execute its jobs. The
organization will work as a service broker, choosing which jobs it
will execute on its own resources and which jobs will have its
execution delegated to other organizations. The set of all jobs
submitted from users of organization \Org{k} is denoted by \Job{k}.

An organization also have costs associated to the execution of the
jobs assigned to its machines. It is, however, free to schedule these
jobs in any way they want, usually to optimize some given objective
(common to all organizations), such as execution cost or global
performance. We denote by $\costlocal{k}$ the cost of the execution of
jobs from \Org{k}'s users when applying the optimal scheduling
according to the organizations' own objective.

Organizations have the freedom to collaborate, organizing themselves
to share their resources, jobs and to redistribute costs.  An
organization that choose to not cooperate must execute all jobs from
its users on its own resources, paying $\costlocal{k}$ to do so.  A
set $\C$ of collaborating organizations share all their machines in
order to execute the jobs in set
$\Job{\C} = \displaystyle \cup_{k\in \C} \Job{k}$.  They are free to
devise a scheduling that optimize their own objective using their
combined resources. We denote as $\costcoop{\C}$ the cost of the
optimal schedule.  For the set $\Job{\C}$, the total price that all
organizations are ready to pay is given by
$\price(\C) = \sum_{k\in \C} \costlocal{k}$.

The payoff of a set of organizations $\C$ is given as a function of
the price charged by each organization to their users
$\price(\C) = \sum_{k\in \C} \costlocal{k} $. The
difference between the sum of the costs that each organization would
have if they were all alone, and the optimal cost achieved when these
organizations are collaborating is known as \emph{utility} $\fctV(\C)$
and is given by:
\begin{equation}
\label{eq:fctV:job}
 \fctV(\C) = 
    \price(\C)   -\costcoop{\C}
 \end{equation}
Using the terminology of  cooperative game theory, $\fctV$ is the \emph{characteristic function} of the game.
Note that if $\C$ is only composed of one organization $k$,  then the utility  is equal to zero ($\fctV(\C) =0)$.

We assume that all organizations have the same objective function. Two
classical objective functions found on the literature are:

\paragraph{($\sum_{J}C_{J}$) \cite{mosp-sumci}:}    Organizations are locally interested in minimizing the average completion time of their jobs.
If scheduling its jobs in its own resources, the cost to schedule \Job{k} in \Org{k}'s own resources is given by
$\costlocal{k}= \underset{ J \in \Job{k}  }\sum (C^{(k)}_{J})$.

\paragraph{($\sum_{J}E_{J}$) \cite{mosp-energia}:}  Each organization \Org{k}, $1 \le k
\le N$, can share several machines that supports continuous \emph{dynamic speed
  scaling} (i.e., processors can operate at any arbitrary
speed $s$ that can be changed by the scheduler over time). A job $J \in \Job{k}$, 
is defined by its release date $r_J^{(k)}=0$, its deadline $d_J^{(k)}$
and its processing volume $1$. The job with the biggest
deadline of \Org{k} is defined as $d^{(k)}_{\max} = \max_i d^{(k)}_i$.
Job preemption is allowed. The energy consumption is given by the integral over
time of the power function $P(s(t)) = s(t)^\alpha$, where $s(t)$ is
the speed in which the processor is running on time $t$ and $\alpha > 1$ is a constant real number that depends on
the technical characteristics of
the processor --- usually $\alpha \in [2,3]$. 
 In other words, an organization \Org{k} can execute its jobs consuming a total
energy of \ekloc{k} only using its own machines.

\begin{example}
\label{ex:mosp-energy}
  Consider an instance of $(\sum_{J}E_{J})$ problem.  Four organizations want to minimize its energy consumption. All these organizations have only one machine. Organizations  \Org{1} and \Org{2} have respectively $19$ and $7$ jobs.  Organizations \Org{3} and  \Org{4}  have $1$ job, and all jobs have the same  deadline date equals to $1$.
So, without cooperation, the local cost of organizations is the following : 
\begin{center}
 $ \costlocal{1} = 19^{\alpha}$,   $\costlocal{2} =7^{\alpha}$ , and   $ \costlocal{3} =\costlocal{4} = 1^{\alpha}$.
\end{center}  
 So each organization \Org{k}  has already to pay $\Cost{\{k\}}(=  p(\{k\}))$ to execute all these jobs.
 Assume that all four organizations form a alliance. This means that all jobs can be executed in all machines.  The optimal schedule is as follow : $7$ jobs are executed on each machine. So,
  the cost ($\Cost{\{1,2,3,4\}}$) of the alliance $\{1,2,3,4\}$  corresponds to the cost of the schedule of all jobs among all machines:   $\Cost{\{1,2,3,4\}} =  ( 4 \cdot 7^{\alpha} ) $. So the energy consumption savings  $\fctV(\{1,2,3,4\})$ is equal to $(19^{\alpha} + 7^{\alpha} + 2 ) - 4 \cdot 7^{\alpha} $. 
  However, organizations \Org{3} and  \Org{4}  increase this own cost, and now their cost is equal to $7^{\alpha}$. So, 
 these organizations should receive some payment  in order their organization to have incentive to take part in the alliance.
  Now, we will give some energy consumption savings of different alliances:
 \[
\begin{array}{ll}
\fctV(\{1,2,3,4\}) = (19^{\alpha} + 7^{\alpha} + 2 ) - 4 \cdot 7^{\alpha}  ;    &\fctV(\{2,3,4\}) = (7^{\alpha} + 2 ) -  3 \cdot 3^{\alpha}  \\
  \fctV(\{1,3,4\})=(19^{\alpha}  + 2 ) - 3 \cdot 7^{\alpha}  ;        &\fctV(\{3,4\}) = 0   \\
      \fctV(\{2,4\})=(7^{\alpha}  + 1^{\alpha} ) - 2 \cdot 4^{\alpha}  ;        &\fctV(\{2,3\}) =(7^{\alpha}  + 1^{\alpha} ) - 2 \cdot 4^{\alpha}  ;     \\

\end{array}
\]

Note that alliances $\{1,2,3,4\}$ and $\{1,3,4\}$ save the same amount of energy.
\end{example}

Let $\CostS{k}$ be the global cost of the cooperative schedule $\ShortestPath$ for organization \Org{k}.
The cooperative problem can then be stated as follows:

\begin{center}
  Find  $(\imputation_{1},\dots, \imputation_{N}$) such that, for all $k$ ($1 \leq k \leq N$), $\CostS{k} - \imputation_{i} \le \costlocal{k}$ \\ if such vector exists.
\end{center}

The vector $\imputation$ represents the payment for each organization to have incentive to collaborate.

In this paper, we will focus on games in which binding   agreements are possible. 
Now, we will focus on  the payoffs of each organization, on the way  which distributes the value of each coalition among its members.

\section{Cooperative game theory \cite{coalitiongame}}


A \emph{cooperative game with transferable utility} is a
pair $(\SSS, \fctV)$ where $\SSS = \{1, \dots, N\}$ is a finite set of
players and a \emph{characteristic function}
$\fctV: 2^{|N|} \to \mathbb{R}$ which associates each subset
$\C \subseteq \SSS$ to a real number $\fctV(\C)$ (such that
$\fctV(\emptyset) = 0$). Each subset $\fctV(\C)$ of $\SSS$ is called a
\emph{coalition}. The function $\fctV$ is a \emph{characteristic
  function} of the game $(\V, \fctV)$ and the \emph{value} of
coalition $\C$ denoted by $\fctV(\C)$ is the value that $\C$ could
obtain if they choose to cooperate. In these games, the value of a
coalition can be redistributed among its members in any possible way.

\subsection{Revenue sharing mechanism}

The challenge of a revenue sharing mechanisms is to find how to split the
payoff $\fctV(\C)$ among the players in $\C$ while ensuring the
stability of the coalition. A vector
$x=(\imputation_1,\dots,\imputation_{|\V|})$ is said to be a
\emph{payoff vector} for a $k$-coalition $\C_1, \dots, \C_k$ if
$\imputation_i \geq 0$ for any $i \in \V$ and
$\sum \limits_{i\in \C_j} \imputation_i \leq \fctV(\C_j)$ for any
$j \in \{1, \dots, k\}$. We will focus on some particular payoff vectors,
namely \textit{imputations}.

\begin{definition}
A payoff vector $\imputation$ for a $k$-coalition $\C_1, \dots, \C_k$ is said to be an \emph{imputation} if it is
\emph{efficient}, --- i.e., $\sum \limits_{i\in \C_j} \imputation_i \leq \fctV(\C_j)$
for any $j \in \{1, \dots, k\}$ --- and if it satisfies the individuality
rationality property --- i.e., $\imputation_i \geq \fctV(\{j\}) $ for any player $j \in \SSS$.
\end{definition}

The objective is to find a \emph{fair}
distribution of the value of the coalition (the payoff of each player
corresponds to his actual contribution to the coalition) and also to
ensure a \emph{stable} coalition in such a way that no player or
subset of players have incentive to leave the coalition.





\paragraph{The objection}

Let $(\SSS, \fctV)$ be a cooperative game and $\imputation$ a payoff
vector of this game.  A pair $(\Objection, y)$ is said to be
\emph{objection of ${i}$ against $j$} if:

\begin{itemize}
\item $\Objection$ is a subset of $\SSS$ such that $i \in \Objection$
  and $j \notin \Objection$ and

\item if $y$ is a vector in $\mathbb{R}^{\SSS}$ such that
  $y(\Objection) \leq \fctV(\Objection)$, for each
  $ k \in \mathcal{P}$, $y_k \geq \imputation_k$ and
  $y_i > \imputation_i$ (agent $i$ strictly benefits from $y$, and the
  other members of $\Objection$ do not do worse in $y$ than in
  $\imputation$).
\end{itemize}


%



\paragraph{The bargaining set.}\cite{Bargainingset,BargainingsetObjection}

A pair $(\ContreObjection, z)$ is said to be a \emph{counter-objection} to an objection
$(\Objection, y)$ if:

\begin{itemize}
\item $\ContreObjection$ is a subset of $\SSS$ such that
  $j \in \ContreObjection$ and $i \notin \ContreObjection$ and

\item if $z$ is a vector in $\mathbb{R}^{\SSS}$ such that
  $z(\Objection) \leq \fctV(\Objection)$, for each
  $k \in \ContreObjection \backslash \Objection$,
  $z_k \geq \imputation_k$ and, for each
  $k \in \ContreObjection \cap \Objection$, $z_k \geq y_k$ (the
  members of $Q$ which are also members of $\Objection$ get at least
  the value promised in the objection).
\end{itemize}

Let $(\SSS, \fctV)$ be a game with a coalition structure. A vector
$x \in \mathbb{R}^{\SSS}$ is \emph{stable} iff for each objection at
$\imputation$ there is a counter-objection.

\begin{definition}
  The bargaining set $\mathcal{B}((\SSS, \fctV))$ of a cooperative game
  $(\SSS, \fctV)$ is the set of stable payoff vectors that are
  \textit{individually rational}, that is,
  $\imputation_i \geq \fctV(\{i\})$.
\end{definition}

Note that it is sufficient to use the notion of imputations since the
payoffs are individually rational.

The bargaining set concept requires objections to be immune to
counter-objections, otherwise they are not considered as credible
threats.

Now, we will go back to our problem and try to define appropriate
characteristic function that reflects the outcome expected from
coalition formation as described earlier in the problem statement
section (Section~\ref{sec:problem-statement}).

\subsection{The choreography enactment pricing game} \label{subsection:choreography}
The choreography enactment game models the cooperative game played by organizations. Their main objective is to form coalitions in order to schedule all jobs belonging to 
them having the lowest cost. 

We consider particular objective functions for a schedule  satisfying the following condition:
 the cost  $(\costcoop{\SSS})$  is the sum of all costs $(\CostS{k})$  of all organizations \Org{k}. Note that this applies to the objective functions $(\sum_{J}E_{J})$  and $(\sum_{J}C_{J})$.

Assume that $\SSS$ should form a coalition. Each organization $\Org{k}$ executes  jobs having
$\CostS{k}$ as a cost. Its cost can be higher than its local cost $\costlocal{k}$. 

In order to incentive organization $\Org{k}$ to join a coalition, the cost savings should be redistributed. So, the  \emph{characteristic function} $\fctV$ of our game corresponds to the cost savings corresponding to Equation~\eqref{eq:fctV:job}.

 Let $\imputation$ be an  imputation. $\Org{k}$ receives $\imputation_{k}$ as payment. So  
$\Org{k}$ has incentive to be in the coalition if $\CostS{k} - \imputation_{i} \le \costlocal{k}$.

Moreover, for the rest of the document, we make two assumptions:
\begin{enumerate}
\item The algorithm building the scheduling assure the \emph{monotonicity property}, i.e,  the cost of scheduling jobs in $\mathcal{J}$ using $m$ machines is less that the cost of scheduling jobs in $\mathcal{J}$ using $m'$ machines if $m<m'$.
\item The total cost for scheduling among all organizations is the sum of the individual costs of each organization.
\end{enumerate}

Note that the costs $(\sum_{J}E_{J})$  and $(\sum_{J}C_{J})$ respect these two constraints. Recall that, without the notion of organizations, they  can be solved in polynomial time and respect the monotonicity property. 
Also note that
the first assumption implies $\fctV(\C) \leq  \fctV(D)$ for every pair of subset $\C, D \subseteq \SSS$ such that $\C\subseteq D$. 
 

\section{Stable coalitions on the choreography enactment pricing game}

This section is devoted to computing the imputation $\imputation$ for the coalition $\SSS$.

Lets focus on organizations that will receive a non-null retribution.
Consider the organizations where participating or not on a coalition
does not alter the amount of savings.

\begin{lemma}
  \label{lem:simple}
  Let $\SSS$ be a set of organizations and $\fctV$ be the characteristic function   (corresponding to the cost savings). Let $\imputation$ be a
  feasible stable imputation. For each organization $ \Org{j}$ in $\SSS$ such that
  $ \fctV(\SSS) = \fctV(\SSS \backslash\{j\})  $, we have $\imputation_j=0$.
\end{lemma}

\begin{proof} Since $\fctV$ respects the monotony constraints: $ \fctV(\SSS) \geq \fctV(\SSS \backslash\{j\})  $ for any organization $ \Org{j}$.

First,  we consider the case where all organizations $ \Org{j}$ are such that  $ \fctV(\SSS) = \fctV(\SSS \backslash\{j\})  $.
It means that  $ \price(\{j\})  =  \costcoop{\SSS} - \costcoop{-j }$ for $\forall j \in  \SSS$. and the cost for $\SSS$ to execute jobs in $\Job{j}$ is equals to  the cost for  organization \Org{j} when it executes its own jobs alone.  Since $ \fctV(\SSS) =   \price(\SSS) - \costcoop{\SSS}$, we have 
$ \fctV(\SSS) =0$ and $\imputation_j=0$ for $\forall j \in  \SSS$.
 
 Second, we consider the case where at least an organization $\Org{i}$ is such that  $ \fctV(\SSS) >  \fctV(\SSS \backslash\{i\})  $. 
 Let $\Org{i^*}$ be such an organization.

 We prove this lemma by contradiction. Assume that there exists one organization $j \in \SSS$ such that
  $ \fctV(\SSS) = \fctV(\SSS \backslash\{j\})  $ and
  $\imputation_j > 0$.

We  prove that $x_{i^{*}}\leq   \fctV(\SSS) - \fctV(\SSS \backslash\{i^{*}\}) $. If it is not the case, then  
organization $\Org{k}\in \SSS$  could make an objection  $(\SSS\backslash\{i^{*}\},y)$
 against organization $\Org{i^{*}}$ such that
 $y_\ell =\imputation_\ell+\frac{\imputation_{i^{*}}-(\fctV(\SSS) - \fctV(\SSS \backslash\{i^{*}\}))}{N-1}$ for
 $\ell \in \SSS\backslash\{{i^{*}}\}$.  Since $\displaystyle \sum_{\ell \in  \SSS\backslash\{{i^{*}}\}} y_{\ell} =  \fctV(\SSS \backslash\{i\}) $, organization $\Org{i^{*}}$ cannot make a  counter-objection against $\Org{k}$.  This means that $\imputation$ is not a
  feasible stable imputation, which leads to a contradiction.

 Therefore, $x_{i^{*}}\leq   \fctV(\SSS) - \fctV(\SSS \backslash\{i\}) $. Organization $\Org{i^{*}}$ could make an objection  $(\SSS\backslash\{j\},y)$
 against organization $\Org{j}$ such that there exists an  $\varepsilon$ such that $   \imputation_j >  \varepsilon  \geq 0$, $y= \imputation_{i^*} +  \varepsilon $, and
 $y_k=\imputation_k+\frac{\imputation_j -\varepsilon}{N-2}$ for
 $k\in \SSS\backslash\{j,i^{*}\}$. Note that for any $k$, $y_k > \imputation_k$ since
 $x_j>0$.

Now lets prove by contradiction that organization $\Org{j}$ cannot make a
 counter-objection $(\ContreObjection,z)$. Assume that organization $\Org{j}$ can make a
 counter-objection $(\ContreObjection,z)$.  Let $\X$ be a set of organizations such that $\X= \SSS\backslash\{j,i\} $. By definition 
 of  counter-objection, we have $\forall k \in \X  z_{k}
\geq y_{k}$ and, $\displaystyle \sum_{k \in \X} z_{k} \geq  \fctV(\SSS) -  (x_{i} + \varepsilon)$.

 Since  $\fctV(\SSS \backslash\{i^*\})= \displaystyle \sum_{k \in \ContreObjection} z_{k}$, we obtain: 
 \[
\begin{array}{lll}
 z_{j}& = &  \fctV(\SSS \backslash\{i^*\}) - \displaystyle \sum_{k \in \X} z_{k} \leq \fctV(\SSS \backslash\{i^*\})  - \fctV(\SSS) +  (\imputation_{i^{*}} + \varepsilon)  \\
    & \leq & -  \imputation_{i^{*}} +  (\imputation_{i^{*}} + \varepsilon) \leq \varepsilon\\
\end{array}
\]

  This contradicts the
 fact with $0\leq  \varepsilon <  \imputation_j \leq  z_{j}$. Thus  $\imputation$ is stable and this conclude the proof.  \hfill \qed
\end{proof}
 
 Now, we focus on participating organizations that does change the amount of savings.

\begin{lemma}\label{lem:6.3}
  Let $\SSS$ be a set of organizations and $\fctV$ be the characteristic function   (corresponding to the cost savings).  Let $\Org{i}$ and  $\Org{j}$ be an organization such that $\fctV(\SSS) > \fctV(\SSS \backslash\{i\})$ and $\fctV(\SSS) > \fctV(\SSS \backslash\{j\})$.  Let
  $\mathcal{O}=\SSS\backslash\{j\}$ be a subset of organizations.

  Let $(\mathcal{O},y)$ be an
  objection of $\Org{i}$ against $\Org{j}$. In order to have a counter-objection
  to $(\ContreObjection, z)$, with
  $\ContreObjection= \SSS\backslash\{i\}$ of $\Org{j}$ against $\Org{i}$, a
  sufficient condition is:
\begin{equation}
  \label{eq:relation}
   \imputation_j -\imputation_i \leq \price(\{j\})   - \price(\{i\}) -   \Cost{-i} 
   +\Cost{-j}
\end{equation}
 \end{lemma}
\begin{proof}
  Assume that there is an objection $(\mathcal{O},y)$ of organization $\Org{i}$
  against organization $\Org{j}$. By definition of objection, we have that
  $\forall k \in \mathcal{O}$, $y_k \geq \imputation_k$ and
  $y_i > \imputation_i$.  Recall that
  $\fctV(\mathcal{O}) =   \price(\mathcal{O})   -\costcoop{\mathcal{O}}$.

Let $\X$ be a set of organizations such that $\X =\SSS \backslash\{i,j\}$.

Now, we will look for a counter-objection of organization $\Org{j}$ using $(\ContreObjection,z)$,  where, for each $k \in \ContreObjection\backslash \mathcal{O}$,
 $z_k \geq \imputation_k$ and for each $k \in \SSS\backslash\{i,j\}$,
 $z_k \geq y_k$. W.l.o.g, we can assume that  $\ContreObjection = \X \cup\{j\}$ and $k \in \X$,
 $z_k = y_k$, otherwise, we can build another counter-objection $\ContreObjection'$  such that $k \in \X$,  $z'_k = y_k$ and $z'_j=z_j +\sum_{k \in \X} (z'_k-  z_k)$.

By combining the definition of characteristic function $\fctV$, we obtain:
 $$\fctV(\ContreObjection)  - \fctV(\mathcal{O})    =  (\price(\ContreObjection)   -\costcoop{\ContreObjection} ) -(   \price(\mathcal{O})   -\costcoop{\mathcal{O}} ) $$
 
By definition of function $\price$, the previous equation can be rewritten

 \begin{equation}\label{eq:inter:1}
\fctV(\ContreObjection)  - \fctV(\mathcal{O})      =     \price(\{j\})   - \price(\{i\})  +\costcoop{\mathcal{O}}   - \costcoop{\ContreObjection}  
  \end{equation}

 Since $\fctV(\mathcal{O}) = \sum_{k\in \X} y_k +y_i$ and
 $\fctV(\ContreObjection) = \sum_{k\in \X} z_k + z_j$,
 Equation~\eqref{eq:inter:1} can be rewritten as:

 \begin{equation}\label{eq:inter:2}
\sum\limits_{k\in \ContreObjection}z_k -    \sum\limits_{k\in  \mathcal{O}}y_k    =
   \price(\{j\})   - \price(\{i\})  +\costcoop{\mathcal{O}}   - \costcoop{\ContreObjection}  \end{equation}

%

  Since $\sum\limits_{k\in \ContreObjection} z_{k}  - \sum\limits_{k\in \mathcal{O}}y_k   =  z_j - y_i
   + \sum\limits_{k\in \X} (y_k - z_k ) $, it is sufficient that:

 \begin{equation}\label{eq:inter:3}
 z_j - y_i
  \leq   \price(\{j\})   - \price(\{i\})   -\costcoop{\ContreObjection}  +\costcoop{\mathcal{O}}  
 \end{equation}
We can notice that $\costcoop{\mathcal{O}}= \costcoop{-j}$ and   $\costcoop{\ContreObjection}= \costcoop{-i}$. 
 From the definition of objection, it is sufficient to have:
 $$\imputation_j -\imputation_i \leq \price(\{j\})   - \price(\{i\}) -   \Cost{-i} 
   +\Cost{-j}.$$
    This concludes the proof of the lemma. \hfill \qed
 \end{proof}


 Lets focus on organizations that will receive non-null
retribution, i.e., organizations on the coalition that does have impact on the amount of cost savings.

%
\begin{theorem}\label{th:5}
Let $\SSS$ be a set of organizations and $\fctV$ be the characteristic function  (corresponding to the cost savings). Let $A$ be a subset of organizations $\{j\in \SSS :  \fctV(\SSS) > \fctV(\SSS \backslash\{j\})\}$. There exists a unique stable imputation $\imputation$ if $\imputation$ fulfills all the three following conditions:
  \begin{enumerate}
  \item $\forall j \in \SSS\backslash A$, $x_{j}=0$
  \item
    $ \displaystyle \forall j \in A$,  $\imputation_j =  \Cost{-j} +\price(\{j\})  - \frac{1}{|A|}\cdot \left( \Cost{A}  + \sum_{k\in A}\Cost{-k}  \right) $; and,
  \item
    $\forall j \in A$, $\Cost{-j} + \price(\{j\})\geq   \frac{1}{|A|}\cdot \left( \Cost{A}  + \sum_{k\in A}\Cost{-k}  \right)   $.
  \end{enumerate}

\end{theorem}
\begin{proof}

Property (1) can be straightforward deduced from Lemma~\ref{lem:simple}.

  The bargaining set is the set of all imputations that do not admit a
  justified objection. So, if we apply Lemma~\ref{lem:6.3} to $i,j$ and
  then to $j,i$, we can derive that for any couple $(i,j) \in \SSS^2$, we
  have $$\imputation_j - \imputation_i =  \price(\{j\})   - \price(\{i\}) -   \Cost{-i} 
   +\Cost{-j}.$$

  Let $\Org{j}$ be an organization in $A$. Summing the previous equations, we obtain:
  \begin{equation}\label{eq:inter:4}
    \sum_{k\in A }\imputation_i -|A|\imputation_j  = \left(\sum_{k\in A}  \price(\{k\}) + \Cost{-k}-|A|(\Cost{-j}+\price(\{j\}))\right)
  \end{equation}

  From the properties of the value of the coalition and by computation, we can rewrite Equation~\eqref{eq:inter:4} as:

  \begin{equation}\label{eq:inter:5}
    \forall j \in A, \; \imputation_j  =   \Cost{-j} +\price(\{j\})  - \frac{1}{|A|}\cdot \left( \Cost{A}  + \sum_{k\in A}\Cost{-k} \right) 
  \end{equation}


  Property (3) can be deduced from the fact that $\imputation_j\geq 0$ and
  from Equation~\eqref{eq:inter:5}. \hfill \qed
\end{proof}

Now, we apply Theorem~\ref{th:5} to Example~\ref{ex:mosp-energy}.  Assume that organizations \Org{2} and  \Org{4} form an alliance ($\SSS=\{2,4\}$). Thus $\Cost{-2}=1^{\alpha}$ and  $\Cost{-4}=7^{\alpha}$. We obtain $\imputation_{2}=\imputation_{4}= \frac{1}{2}(7^{\alpha}  + 1^{\alpha} ) -  4^{\alpha} $. 
This means that in this alliance organization \Org{4} pays $\price({\{4\}})$. From the alliance, it receive two payments. The first payment is due to the cost executing tasks assigned to it from other organizations in the alliance, and the second corresponds to the reward $\imputation_{4}$  for taking part in it.
In any case, the imputation represents either a reduction in the costs of overloaded organizations (in this example, organization \Org{2}) or is a payment for the use of these resources by another organization (in this example, organization \Org{4}).

Now we will establish the relation between the price and the different costs when there is a stable imputation. Using Theorem~\ref{th:5}, we will compute the lower bound for non-empty bargaining sets.

\begin{corollary}\label{th:6}
Let $\SSS$ be a set of organizations and $\fctV$ be the characteristic function $\fctV$ (corresponding to the cost savings). Let $A$ be a subset of organizations $\{j\in \SSS :  \fctV(\SSS) > \fctV(\SSS \backslash\{j\})\}$. There exists a unique stable imputation $\imputation$ if $ \price(A)\geq  \Cost{A}    $.
 \end{corollary}

\begin{proof}
From Theorem~\ref{th:5}, we have:
\begin{center}
$\forall j \in A$, $\Cost{-j} + \price(\{j\})\geq   \frac{1}{|A|}\cdot \left( \Cost{A}  + \sum_{k\in A}\Cost{-k}  \right)   $.
\end{center}
By summing these $|A|$ inequations, we have:
\begin{center}
 $\displaystyle\sum_{k\in A} \left( \Cost{-j} + \price(\{j\}\right)\geq  \Cost{A}  + \sum_{k\in A}\Cost{-k}    $.
\end{center}
The previous inequation can then be rewritten as:
 $ \price(A)\geq  \Cost{A}$. \hfill $\qed$
\end{proof}

These results constitutes the mathematical framework needed to devise an algorithm that, given a set of organizations $\SSS$, determines whether an alliance is possible and, if possible, computes the imputation.





\begin{corollary}\label{th:7}
Let $\SSS$ be a set of organizations with their sets of jobs. 
  If all organizations have as objective function $(\sum_{J}C_{J})$ or $(\sum_{J}E_{J})$, then Algorithm~1 determines in polynomial time whether $\SSS$ can form a coalition and, if possible, returns the imputation vector.
\end{corollary}

\begin{proof}
  These objective functions respect the monotonicity property and an
  scheduling minimizing these objectives can be devise in polynomial
  time (see \cite{albers2011multi} for a polynomial algorithm for
  $(\sum_{J}E_{J})$). Therefore, we can apply Theorem~\ref{th:5} and
  Corollary~\ref{th:6}, which are implemented by Algorithm~1. \hfill
  \qed
\end{proof}

Note that the same result straightforwardly applies to any scheduling
problem whose objective function respects the two assumptions
described in Section~\ref{subsection:choreography}), as long as they
can be solved in polynomial time.

\begin{algorithm}[ht]
\KwIn{ A set $\SSS$ of organizations,  function $\fctV$   (corresponding to the cost savings), and  $\Cost{.}$.}
\KwOut{(Whether there is a alliance or not and the imputation vector)}

Compute the lowest cost schedule $\ShortestPath$\;

\ForAll{organizations $\Org{\player} \in \SSS$}{Compute the lowest cost schedule using only its own resources and its local cost $\left(\costlocal{\player}=\price(\{\player\})\right)$\; 

Compute the lowest cost schedule using  all  resources except $\Org{\player}$'s resources and  its  cost $\left(\Cost{-\player}\right)$\; 

}

Compute the lowest cost schedule using  all the resources and  its  cost ($\Cost{\SSS}$)\; 

\ForAll{organizations $\Org{\player} \in \SSS$}{Compute $ \price({\SSS\backslash\{\player\}})$ $\left(=\sum_{j\in \SSS, j\neq \player} \price(\{j\})\right)$\; 
Compute $ \fctV({\SSS\backslash\{\player\}})$ $\left(=\price({\SSS\backslash\{\player\}}) - \Cost{-\player}\right)$\;  
}

Compute $A= \{j\in \SSS \mid  \fctV(\SSS) > \fctV(\SSS \backslash\{j\})\}$, $ \price(A)$,  $\Cost{A}$   and
$\displaystyle \sum_{k\in A}\Cost{-k} $\;

 \If{$\price(A)<  \Cost{A}$ }{%
 \Return{(coalition=false, imputation=$\emptyset$)}}%

\ForAll{organization $\Org{\player} \in A$}{
  compute $\imputation_\player$ according to Equation (2) of Theorem~\ref{th:5}\;
}

\Return{(coalition=true, imputation=$\imputation$)}

\caption{Coalition detection algorithm.}
\label{algo:detection1}
\end{algorithm}

\section{Related work}

The problem of scheduling jobs on independent, selfish organizations
sharing a common infrastructure ---known as the Multi Organization
Scheduling Problem or MOSP--- was first studied by Pascual et
al.~\cite{pascual07,pascual11} and was then extended to include
notions from game-theory by Cohen et
al.~\cite{mosp-europar,cpe15}. The original studies does not include
the ability to form coalitions; MOSP only allows rebalancing the
jobs between organizations, as long as no organization presents a
performance degradation according to their its own performance
objective.

The choreography enactment pricing game is similar to fair resource
allocation and networking games. Fragnelli et al.~\cite{FragnelliGM00}
studied a related cooperative game that they called the shortest path
games. Their game models agents willing to transport a good through a
network from a source to a destination. Using a graph model, and
letting agents to control the nodes, they have studied how profits
should be allocated according to the core of the cooperation.


Maintaining the assumption of $s$-veto players, Voorneveld and
Grahn~\cite{VoorneveldG02} extended the shortest path game and proved
that the core allocations coincide with the payoff vectors in the
strong Nash equilibria of the associated non-cooperative shortest path
game.

Several subsequent papers~\cite{Aziz2011,Nebel11} studied
computability and complexity aspects of this game.  Some
properties of graphs and games guaranteeing the existence of a core
have been proposed and the computability complexity of computing cores
have been established (NP-complete and \#P-complete). Other variants
(different payoffs and players controlling arcs) have been considered,
but mostly focusing on the existence and complexity of cores, whereas
this work mostly focus on the construction of the bargaining set in
polynomial time.

The flow game can be view as maximum multicommodity flow problem in a
cooperative setting. This
model can be used to identify the set of demands to satisfy and to
route this demand on the network. In this context, players own network
resources and share a capacity to deliver commodities. Kalai et
al.~\cite{Kalai1982} first considered flow games for network with a
single commodity, where a unique player owns an arc. Several studies
(for example, \cite{Derks1985,MarkakisS03,AgarwalE08}) extended this
seminal work. Those extensions encompass variations on the number of
arcs a player can control, if the player controls all or a part of the
capacity of the arc, if players control vertices, etc. Those papers
mainly focus on how to obtain the optimal flow in the network and then
on how to allocate the revenue using core allocation techniques (since
those games have non-empty cores).


\section{Conclusion}
This work presents a game-theoretic model for the problem we call the
choreography enactment pricing problem. Service providers must enact,
i.e., assign all jobs that composes the users' applications (service
compositions) to resources of its own or subcontract resources from
other service providers.
Each organization must pay a price to execute the jobs from its
users. This price depends on a given cost function that is common to
all organizations (e.g., energy costs). In order to reduce its costs,
a set of organizations can form alliances (coalitions) to share its
resources and reduce its costs.

From an economic point of view, the lack of regulations may create
alliances that can control the prices and/or limit competition. Our
study of the choreography enactment pricing problem as a cooperative
game characterized the mathematical conditions needed for the creation
of stable coalitions. From this characterization we devised an
algorithm to detect if such alliances can be formed.

The study of the bargaining set of this problem suggests that this
problem may be related to the notion of \emph{truthful mechanism} from
algorithmic mechanism design in game theory. To the best of our
knowledge, there is no study about this relation and it would be an
interesting future work.

Also as future work, we will consider this problem with a broader set of
cost functions (such as the total makespan). For now,
our results cannot be apply cost functions that does not respect
monotony constraints.  One of the first step is to consider the load
as the cost (i.e. the sum of all jobs executed in the same
organization). Recently, Azar~ \cite{AzarHMRV15} designed a polynomial
randomized $2$-approximation algorithm for minimizing makespan using
restricted-related machines. They use technique to find an optimal
fractional solution\cite{Lenstra:1990} and they modify optimal
fractional solution that the load assigned to machine is monotone.
Using this result and from the expectation linearity, all our results
can be applied. It means that to find an imputation $\imputation$ such
that, for all $k$ ($1 \leq k \leq N$),
$\mathbb{E}[\CostS{k}] - \imputation_{i} \le
\mathbb{E}[\costlocal{k}]$ if such vector exists can be computed in
polynomial time. This technique can may be help to solve for the
formation of coalition considering makespan as the cost.

\end{document}